\documentclass[11pt,a4paper,longbibliography]{amsart}
\usepackage[margin=2.5cm]{geometry}
\usepackage{cite}
\usepackage{amsmath,amsfonts,amssymb,amsthm,amscd,mathrsfs,mathdots,bbm,color,bm,enumerate,comment,graphicx,newtxmath,newtxtext,subcaption,stmaryrd,caption,soul, xcolor}
\usepackage{comment}
\usepackage{physics}
\usepackage{multirow}
\usepackage[normalem]{ulem}

\numberwithin{equation}{section}
\definecolor{equationcolor}{RGB}{222,94,100}
\definecolor{alecolor}{RGB}{198,113,190}

\usepackage[linecolor=white,backgroundcolor=orange,bordercolor=red]{todonotes}

\usepackage[colorlinks=true,linkcolor=equationcolor,citecolor=teal]{hyperref}

\hypersetup{urlcolor=teal}
\definecolor{equationcolor}{RGB}{222,94,100}

\def\R{{\mathbb R}}

\def\R{{\mathbb R}}

\def\e{\mathbbm{e}}

\usepackage{thmtools,amsthm}
\theoremstyle{plain}

\newtheorem{thm}{Theorem}[section]

\newtheorem{prop}[thm]{Proposition}

\newtheorem{defn}{Definition}[section]
\newtheorem{ex}{Example}[section]

\theoremstyle{definition}

\theoremstyle{remark}
\newtheorem{rem}{Remark}

\usepackage[utf8]{inputenc}

\def\_#1{\def\next{#1}
	\ifx\next\risingsign\expandafter\rising\else^{\underline{#1}}\fi}
\def\risingsign{^}
\def\rising#1{^{\overline{#1}}}

\title[Limits of  Stochastic Semigroups and Block-Triangular Majorisation]{Limits of Stochastic Semigroups and Block-Triangular Majorisation}

\author[F.~D.~Cunden]{Fabio Deelan Cunden}
\address{Dipartimento di Matematica, Universit\'a degli Studi di Bari, I-70125 Bari, Italy, and
	INFN, Sezione di Bari, I-70126 Bari, Italy}
\email{fabio.cunden@uniba.it}

\author[J. Czartowski]{Jakub Czartowski}
\address{School of Physics, Trinity College Dublin, Dublin 2, Ireland}
\address{School of Physical and Mathematical Sciences, Nanyang Technological University, 21 Nanyang Link, 637371 Singapore, Republic of Singapore}
\email{jakub.czartowski@tcd.ie}

\author[G.~Gramegna]{Giovanni Gramegna}
\address{Dipartimento di Fisica, Universit\`a degli Studi di Bari, I-70126 Bari, Italy, and
	INFN, Sezione di Bari, I-70126 Bari, Italy}
\email{giovanni.gramegna@uniba.it}

\author[M. Ligab\`o]{Marilena Ligab\`o}
\address{Dipartimento di Matematica, Universit\'a degli Studi di Bari, I-70125 Bari, Italy}
\email{marilena.ligabo@uniba.it}

\allowdisplaybreaks

\begin{document}
	
	\begin{abstract}  
		We investigate limits of semigroups of stochastic matrices defined by their invariant distribution. Given probability vectors $q(\beta)$ depending on a parameter $\beta$, we introduce a notion of convergence as $\beta\to\infty$ for the corresponding semigroups of $q(\beta)$-preserving stochastic matrices and investigate the structure of the resulting limit. In general, the limiting semigroup differs from the semigroup preserving the limiting distribution, showing that these two operations do not commute.
		We develop a general framework for such limiting semigroups and study in detail the case in which the invariant distributions are Gibbs vectors at  inverse temperature $\beta$. In that case, the limiting semigroup consists of block-upper-triangular stochastic matrices subject to additional substochasticity constraints. We characterise and enumerate their extremal elements and determine the preorder on probability vectors induced by the action of the semigroup.
		The resulting notion of Block-Triangular majorisation interpolates between ordinary majorisation and upper triangular (also known as unordered) majorisation. We show that it is completely characterised by a finite family of monotones and analyse the corresponding behaviour of R{\'e}nyi $\alpha$-entropies as $\beta\to\infty$.
	\end{abstract}
	
	\maketitle

	\section{Introduction}
	\label{sec:intro}
	Majorisation governs a number of fundamental state-conversion problems in quantum information theory and statistical physics. Examples include bipartite pure-state entanglement transformations~\cite{nielsen1999conditions}, coherence manipulation~\cite{baumgratz2014quantifying}, and thermodynamic state transitions under Gibbs-preserving constraints~\cite{horodecki2013fundamental}. Among these, the resource theory of athermality exhibits the richest majorisation structure, described by thermomajorisation~\cite{gour2015resource,deoliveirajunior2022}.
	
	While much recent work has focused on coherent systems, especially in asymptotic \cite{shiraishi2025recovery} and catalytic settings \cite{son2024hierarchy}, the energy-incoherent regime continues to reveal nontrivial structural aspects of finite-size thermodynamics. In particular, degeneracies in the energy spectrum introduce additional geometric and combinatorial features that are not present in the nondegenerate case.
	
	For a given nonincreasing sequence of energy levels $E_1\geq E_2\geq\cdots \geq E_n$, consider the associated Gibbs probability vector at each inverse temperature $\beta\geq0$,
	\begin{equation}
		\label{eq:Gibbs_intro}
		q(\beta)
		=
		\frac1{Z(\beta)}
		\bigl(
		e^{-\beta E_1},
		\ldots,
		e^{-\beta E_n}
		\bigr),
		\qquad
		Z(\beta)=\sum_{i} e^{-\beta E_i}.
	\end{equation}
	Let $ \mathcal{S}^{q(\beta)}$ denote the set of $n\times n$ row-stochastic matrices $P$ satisfying
	$
	q(\beta)P=q(\beta)
	$.
	Under matrix multiplication, $\mathcal{S}^{q(\beta)}$ is a semigroup with identity. Equivalently, it is the semigroup of Markov transition matrices having $q(\beta)$ as a stationary distribution. The convex geometry of this set, its extreme points and its relation to other notions such as transportation polytopes has been extensively studied in the literature~\cite{Hartfiel1974, GregoryKirklandPullman1992, vom2019d,Cunden2025SemigroupMajorisation}.
	
	A natural question is: what becomes of the finite-temperature semigroups
	$\mathcal{S}^{q(\beta)}$
	in the zero-temperature limit, $\beta\to\infty$? 
	As the temperature decreases, the Gibbs distribution $q(\beta)$ becomes uniformly supported in the lowest energy space. Let $q(\infty)$ be this limit. If we denote by $\mathcal{S}^{q(\infty)} $ the set of $q(\infty)$-preserving row-stochastic matrices, is it true that
	\begin{equation}
		\label{eq:question}
		\mathcal{S}^{q(\beta)} \to \mathcal{S}^{q(\infty)}, \quad \beta\to\infty?
	\end{equation}
	At low temperatures, the Gibbs state becomes increasingly concentrated on the ground-state sector. Since absolute zero is an ideal limit, it is natural to study the behaviour of the semigroups $ \mathcal{S}^{q(\beta)}$ as $\beta\to\infty$. Equation~\eqref{eq:question} asks whether taking the zero-temperature limit commutes with passing from a Gibbs state to the semigroup of stochastic maps that preserve its limit $q(\infty)$.
	
	In this paper we show that the answer is negative: taking the limit of invariant distributions and taking the semigroup of stochastic matrices preserving them are, in general, non-commuting operations. Different families of invariant (not necessarily Gibbs) distribution $q(\beta)$ may converge to the same limiting distribution while giving rise to different limiting semigroups. Understanding precisely which features of the asymptotic behaviour determine the limit, and which are irrelevant, is the driving question of the present work. In this paper, we first clarify the precise meaning of the limit on the left-hand side of~\eqref{eq:question}, that is, the limit of a family of semigroups of matrices indexed by the parameter $\beta\ge0$, and then compute the limit.
	
	\subsection{Outline of the paper}
	\label{sec:outline}
	The paper is organised as follows. We start in Section~\ref{sec:preparatory} by setting the notation, terminology and definitions on stochastic matrices, semigroup majorisation, and Kuratowski convergence for sequences of sets. In Section~\ref{sec:results} we present the results.  The proofs are contained in Section~\ref{sec:proofs}.
	\section{Preliminaries}
	\label{sec:preparatory}
	
	\subsection{Stochastic matrices}
	
	We briefly recall the required notions from majorisation theory; a standard reference is~\cite{marshall1979inequalities}.
	We write $\R_+$ and $\R_{++}$ for the sets of nonnegative and strictly positive real numbers, respectively. Let $\e=(1,\dots,1)\in\R^n$.
	

	
	An $n\times n$ matrix $P$ with nonnegative entries is called
	\emph{(row-)stochastic} if each row sums to one, that is,
	\begin{equation}
		P\e^T=\e^T.
	\end{equation}
	If the row sums are at most one, so that $P\e^T\leq\e^T$,
	then $P$ is called \emph{row-substochastic}.
	Throughout, inequalities between vectors are understood componentwise.
	
	Given $q\in\R_+^n\setminus\{0\}$, a nonnegative matrix $P$ is called
	\emph{$q$-preserving} if
	\begin{equation}
		qP=q,
	\end{equation}
	and \emph{$q$-subpreserving} if $qP\leq q$.
	For $q=\e$, the latter condition means that all column sums are
	at most one; such a matrix is called \emph{column-substochastic}.
	A matrix that is both row- and column-substochastic is called
	\emph{doubly substochastic}.
	For row-stochastic matrices, normalization implies that $q$-subpreservation is equivalent
	to $q$-preservation.
	
	We denote by $\mathcal{S}$ the semigroup of $n\times n$
	stochastic matrices and by $\mathcal{S}^q$ its subsemigroup
	of $q$-preserving matrices.
	Since the identity matrix belongs to $\mathcal{S}^q$,
	this set is nonempty.
	If $q$ is normalized with respect to the 1-norm, the rank-one matrix $\e^Tq$,
	whose rows all equal $q$, also belongs to $\mathcal{S}^q$.
	The set $\mathcal{S}^q$ is a compact convex polytope
	in $\R^{n\times n}$ of dimension at most $(n-1)^2$,
	with equality whenever $q\in\R_{++}^n$.
	
	\subsection{Semigroup majorisation}
	
	Let
	\begin{equation}
		\Delta_{n-1}
		:=
		\left\{
		x\in\R_+^n :
		\sum_{i=1}^{n} x_i =1
		\right\}
	\end{equation}
	denote the probability simplex.
	Every stochastic matrix acts on $\Delta_{n-1}$ by right multiplication. Let $\mathcal{T}\subseteq\mathcal{S}$ be a subsemigroup of $\mathcal{S}$.
	
	\begin{defn}[Semigroup majorisation]
		For $x,y\in\Delta_{n-1}$, we say that $x$ is $\mathcal{T}$-majorised by $y$, and write
		$  x\prec^{\mathcal T} y$,
		if there exists $P\in\mathcal T$ such that $  x=yP$.
	\end{defn}
	
	If $\mathcal T$ is the semigroup of doubly stochastic matrices, one recovers ordinary majorisation:
	\begin{equation}
		x\prec y
		\iff
		\sum_{j=1}^k x_j^\downarrow
		\leq
		\sum_{j=1}^k y_j^\downarrow,
		\qquad
		k=1,\dots,n,
	\end{equation}
	where $x^\downarrow$ denotes the nonincreasing rearrangement of $x$, i.e. $x_1^\downarrow\geq x_2^\downarrow \geq \dots \geq x_n^\downarrow$.
	
	If $\mathcal T$ is the semigroup of upper triangular stochastic matrices, one obtains upper triangular (UT) majorisation, also called unordered majorisation:
	\begin{equation}
		x\prec^{\mathrm{UT}} y
		\iff
		\sum_{j=1}^k x_j
		\leq
		\sum_{j=1}^k y_j,
		\qquad
		k=1,\dots,n.
	\end{equation}
	
	Finally, if $\mathcal T= \mathcal{S}^q$, we write $    x\prec^q y$.
	When $q$ is a Gibbs probability vector, this relation is known as thermomajorisation.  See~\cite{horodecki2013fundamental, vom2019d} for more details on the structure of $q$-majorisation and its appearance in the context of resource theory of athermality.
	
	\subsection{Kuratowski convergence}
	As we will be concerned with the low-temperature limits of semigroups, we recall the appropriate notions of lower and upper set limits following Painlev\'e and Kuratowski~\cite[Chap.~29]{KuratowskiI}.

	Let $(A_k)_{k\geq1}$ be a sequence of subsets of a metric space $(X,d)$.
	\begin{defn}
		The point $p$ belongs to the \textbf{lower limit} $\operatorname{Li} A_k$ of the sequence of sets $(A_k)_{k\geq1}$,  if every neighbourhood of $p$ intersects $A_k$ for all but finitely many $k$. 
	\end{defn}
	\begin{defn}
		The point $p$ belongs to the \textbf{upper limit} $\operatorname{Ls} A_k$ of the sequence of sets $(A_k)_{k\geq1}$, if every neighbourhood of $p$ intersects $A_k$ for infinitely many $k$. 
	\end{defn}
	
	Equivalently, $   p\in \operatorname{Li} A_k$ if and only if there exist points $p_k\in A_k$ such that $p_k\to p$. 
	Similarly, $p\in \operatorname{Ls} A_k$ 
	if and only if there exists a subsequence $(m_k)$ and points $p_{m_k}\in A_{m_k}$ such that $p_{m_k}\to p$.
	Elements of $\operatorname{Li}A_k$ are called `limit points' of  $(A_k)_{k\geq1}$, and elements of $\operatorname{Ls}A_k$ are called `cluster points' of $(A_k)_{k\geq1}$.
	The lower and upper limits are closed sets, and one always has
	\begin{equation}
		\operatorname{Li} A_k
		\subseteq
		\operatorname{Ls} A_k.
	\end{equation}
	
	\begin{defn}
		The sequence $(A_k)_{k\geq1}$ is said to be \textbf{convergent} to $A$ in the sense of Kuratowski if $\operatorname{Li} A_k=A =\operatorname{Ls} A_k$. We then write $A=\operatorname{Lim}A_k$.
	\end{defn}
	The preceding definitions extend verbatim to nets $(A_i)_{i\in I}$. In this paper, we shall only consider families of sets in $X=\mathcal{S}$ indexed by positive reals $\beta\geq 0$, and upper and lower limits are always understood in the sequential sense as $\beta\to\infty$. More precisely, for a family $(A_\beta)_{\beta\geq0}$, we write $\operatorname{Ls}A_\beta$ for the set $\operatorname{Ls}A_{\beta_k}$ whenever the latter is independent of the choice of sequence $(\beta_k)_{k\geq1}$ satisfying $\beta_k\to\infty$, as $k\to\infty$. The same convention applies to $\operatorname{Li}A_\beta$.
	
	Let $\mathcal F(X)$ denote the set of closed subsets of $X$, the limit being understood in the sense of Kuratowski. The space $\mathcal{F}(X)$ may fail to be topological. However, if $(X,d)$ is compact, then $\mathcal F(X)$ is metrizable  (see, e.g.,~\cite{Burago2001}) via the Hausdorff distance $d_H$, defined as
	\begin{equation}
		d_H(A,B):= \max\left\{\sup_{a\in A} d(a,B),\sup_{b\in B} d(A,b)\right\},\quad \text{for all $A,B\in \mathcal F(X)$.}
	\end{equation}
	In this case,
	\begin{equation}
		\operatorname{Lim}A_k = A
		\iff
		d_H(A_k,A)\to0.
	\end{equation}
	\begin{defn} Let $(X,d)$, $(Y,d')$ be compact metric spaces. 
		A map
		\begin{equation}
			f:(Y,d')\to (\mathcal F(X),d_H)
		\end{equation}
		is said to be \textbf{Painlev\'e--Kuratowski continuous} at $y\in Y$ if
		\begin{equation}
			d'(y_n, y)\to 0
			\quad\Longrightarrow\quad
			d_H\bigl(f(y_n),f(y)\bigr)\to0.
		\end{equation}
		Equivalently,  $y_n\to y$ implies     $f(y_n)\to f(y)$,
		in the sense of Kuratowski.
	\end{defn}

	\section{Results}
	\label{sec:results} 
	We consider the case of a continuous curve   
	$$
	\R_+ \to \Delta_{n-1},\quad \beta\longmapsto q(\beta),
	$$
	converging as $\beta\to\infty$ to a boundary point of $\Delta_{n-1}$.   The probability vectors $q(\beta)$  may not necessarily be Gibbs vectors. For definiteness, we  assume the components of $q(\beta)$ to be nondecreasing for sufficiently large $\beta$. Theorem~\ref{thm:general}  presents the Kuratowski limit as $\beta\to\infty$ of the corresponding  $q(\beta)$-preserving semigroups $\mathcal{S}^{q(\beta)}$. In particular, the limit depends on the multiscale profile of $q(\beta)$ (see Sec.~\ref{sec:multiscale} below), and this shows that the map $q\longmapsto \mathcal{S}^q$ is not Painlevé--Kuratowski continuous at some boundary points of $\Delta_{n-1}$. Figure~\ref{fig:cartoon} illustrates the setting of Theorem~\ref{thm:general} in dimension $n=3$.
	
	\begin{figure}
		\centering		\includegraphics[width=.9\linewidth]{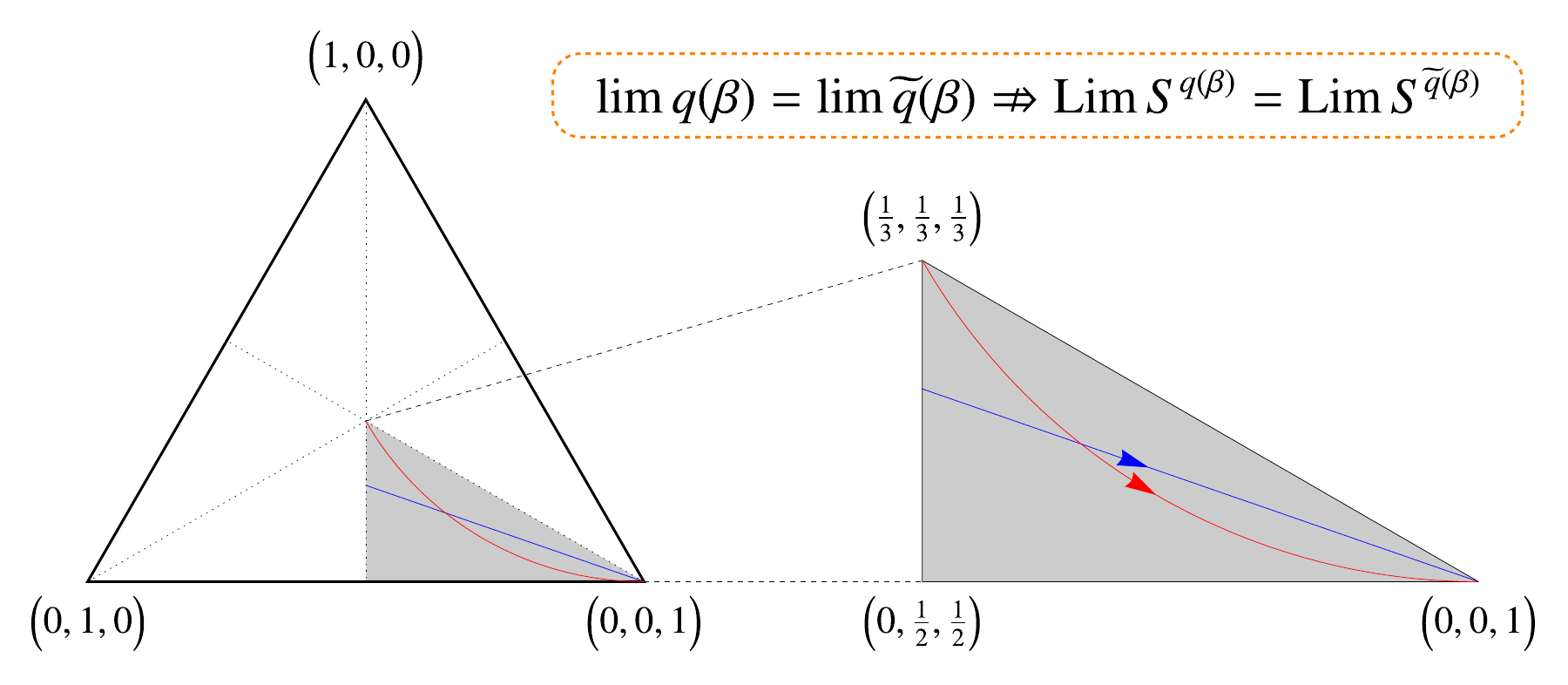}
		\caption{Dependence of the limiting distribution-preserving semigroup on the multiscale profile. The two families of probability vectors 
			$q(\beta)\propto (\tfrac12 e^{-\beta},e^{-\beta},1)$, and $q(\beta)\propto (e^{-2\beta},e^{-\beta},1)$, shown in blue and red, respectively, converge to the same boundary point $(0,0,1)\in\Delta_2$ as $\beta\to\infty$, but possess different multiscale asymptotic structures. Consequently, the associated semigroups of $q$-preserving stochastic matrices have distinct Kuratowski limits. This demonstrates that the limiting semigroup is determined by the full multiscale profile rather than by the limiting distribution alone.}
		\label{fig:cartoon}
	\end{figure}
	Theorem~\ref{thm:main} is the specialisation of Theorem~\ref{thm:general}  to the $\beta\to\infty$ limit of $\mathcal{S}^{q(\beta)}$, when $q(\beta)$ is a Gibbs vector with arbitrary multiplicities of the energies. In this case the limit semigroup only depends on the degeneracy structure of $q(\beta)$, and it is a semigroup of block-upper-triangular matrices that we dub BT-stochastic matrices (`BT' for `block triangular'). We  identify and enumerate the extremal BT-stochastic matrices  (Proposition~\ref{prop:extreme_pts}). We then characterise the induced semigroup majorisation in terms of monotones (Theorem~\ref{prop:monotones}), and we study the zero-temperature limit of the R\`enyi entropic inequalities (Proposition~\ref{prop:fate_monotones}). The limiting semigroups and their extreme points are illustrated in Section~\ref{sec:limit_sets} for $n=3$ and $n=4$.
	
	\subsection{Multiscale profiles and associated limiting stochastic matrices}
	\label{sec:multiscale}
	A \textbf{composition} of $n$ is a sequence of positive integers
	$\mu=(m_1,\ldots,m_k)
	$
	such that $m_1+\cdots+m_k=n$.
	We call $\mu$ a $k$-composition when exactly $k$ parts appear. There are $ 
	\binom{n-1}{k-1}
	$
	such compositions, and altogether $2^{n-1}$ compositions of $n$. See~\cite{Stanley2011}.
	In this context, compositions encode the vanishing rates of the components of $q(\beta)$, reflecting the degeneracy structure of the energies for Gibbs vectors.
	
	Let $(q(\beta))_{\beta\geq0}$ be a family of probability vectors
	whose components are strictly positive for all sufficiently large
	$\beta$.
	\begin{defn}
		Given a composition $\mu=(m_1,\ldots,m_k)$ of $n$, we say that
		$q(\beta)$ admits a \emph{multiscale decomposition of type $\mu$}
		if it can be decomposed into consecutive blocks
		\[
		q(\beta)=\bigl(q^1(\beta),\ldots,q^k(\beta)\bigr),
		\qquad q^a(\beta)\in\R_+^{m_a},
		\]
		with the following properties:
		\begin{itemize}
			\item [i. ] Within each block, the components are asymptotically proportional:
			\begin{equation}\label{eq:convergence-IntraBlock}
				\lim_{\beta\to\infty}
				\frac{q_i^a(\beta)}{q_j^a(\beta)}
				\in(0,\infty),
				\qquad
				i,j=1,\ldots,m_a,\quad a=1,\ldots,k.
			\end{equation}
			\item[ii.] The blocks are ordered by increasing asymptotic scale and are
			separated in the sense that
			\begin{equation}
				\frac{\|q^a(\beta)\|_1}{\|q^{a+1}(\beta)\|_1}
				\longrightarrow0,
				\qquad a=1,\ldots,k-1.
			\end{equation}
		\end{itemize}
		Condition $\mathrm{i.}$ within each block guarantees that the normalized
		vectors converge:
		\[
		\frac{q^a(\beta)}{\|q^a(\beta)\|_1}
		\longrightarrow q^a
		\in\Delta_{m_a-1}\cap\R_{++}^{m_a},
		\qquad a=1,\ldots,k.
		\]
		We call the tuple $(q^1,\ldots,q^k)$ the \emph{multiscale profile}
		of the family $(q(\beta))_{\beta\geq0}$.
	\end{defn}
	Informally, $q(\beta)$ is of type $\mu=(m_1,\dots,m_k)$ if
	$$
	\|q^1(\beta)\|_1
	\ll
	\|q^2(\beta)\|_1
	\ll\cdots\ll
	\|q^k(\beta)\|_1,
	\qquad
	\beta\to\infty,
	$$
	
	To describe the $\beta\to\infty$ limit of $\mathcal{S}^{q(\beta)}$, it is convenient to organise matrices according to the type of $q(\beta)$.
	We thus partition an $n\times n$ matrix into blocks,
	$$
	P=\bigl(P^{ab}\bigr)_{a,b=1}^k,
	$$
	where the block $P^{ab}$ has size $m_a\times m_b$. Indices inside the blocks will rarely matter and will be suppressed whenever possible.

	Suppose furthermore that, for each $a=1,\ldots,k$,
	$$
	\frac{q^a(\beta)}
	{\|q^a(\beta)\|_1}
	\longrightarrow
	q^a, \qquad \beta \to \infty,
	$$
	for some probability vector $q^a\in\Delta_{m_a-1}$.
	We call the tuple $(q^1,\ldots,q^k)$ 
	the \textbf{multiscale profile} of the family $\left(q(\beta)\right)_{\beta\geq0}$. As $\beta\to\infty$, 
	\begin{equation}
		\label{eq:limit_gamma_gen}
		q(\beta)
		\to
		q(\infty)
		=
		\big(
		\underbrace{0,\ldots,0}_{n-m_k},
		q^{k}_1,\ldots,q^{k}_{m_k}
		\big).
	\end{equation}
	\begin{ex} 
		Let $\beta>0$ and $$q(\beta)=c_{\beta}\left(e^{-\beta}, 2e^{-\beta},\frac{1}{\beta},\frac{1}{\beta},1,2,3\right),$$ where $c_{\beta}=(3e^{-\beta}+\frac{2}{\beta}+6)^{-1}$. This is a vector of type $\mu=(2,2,3)$, as $\beta\to\infty$. Its multiscale profile is 
		$$ 
		q^1=\left(\frac{1}{3},\frac{2}{3}\right),\qquad
		q^2=\left(\frac{1}{2},\frac{1}{2}\right),\qquad
		q^3=\left(\frac{1}{6},\frac{1}{3},\frac{1}{2}\right).
		$$
	\end{ex}
	The semigroup $\mathcal{S}^{q(\infty)}$ preserving the limiting  vector~\eqref{eq:limit_gamma_gen}
	consists of row-stochastic matrices with blocks
	\begin{equation}
		P^{k1}=0, P^{k2}=0,\ldots,  P^{k,{k-1}}=0  
	\end{equation}
	while the lower-right $m_k\times m_k$ block $P^{kk}$ is $q^k$-preserving.
	
	Back to our question: \emph{does $\mathcal{S}^{q(\beta)}$ converge to $\mathcal{S}^{q(\infty)}$, or does a more rigid structure emerge in the limit?}
	
	It turns out that transitions against the blocks of $q(\beta)$ that are separated in scale are suppressed as $\beta\to\infty$. What survives in the limit is therefore not the full  semigroup $\mathcal{S}^{q(\infty)}$, but a subsemigroup of it with a block-upper-triangular structure.
	\begin{defn}
		Let $q^a\in\Delta_{m_a-1}$ for $a=1,\ldots,k$. Define $\mathcal T^{q^1,\ldots,q^k}$ as the semigroup of row-stochastic block-upper-triangular matrices 
		\begin{equation}
			P=
			\begin{pmatrix}
				P^{11} & P^{12} & \cdots & \cdots & P^{1k}\\
				0 & P^{22} & \cdots & & \vdots\\
				\vdots & \vdots & \ddots & & \vdots\\
				\vdots & & & P^{k-1,k-1} & P^{k-1,k}\\
				0 & \cdots & \cdots & 0 & P^{kk}
			\end{pmatrix},
		\end{equation}
		where each block $P^{ab}$ has size $m_a\times m_b$, and the diagonal blocks $P^{aa}$ are $q^a$-subpreserving:
		\begin{equation}
			q^aP^{aa}\leq q^a
		\end{equation}
		for all $a=1,\ldots,k$. (In particular, the last block $P^{kk}$ is $q^k$-preserving.)
	\end{defn}
	\begin{thm}
		\label{thm:general}
		Let $\left(q(\beta)\right)_{\beta\geq0}$ be a family of probability vectors of type $\mu=(m_1,\ldots,m_k)$ as $\beta\to\infty$, with multiscale profile $(q^1,\ldots,q^k)$. Then,
		\begin{equation}
			\operatorname{Lim}\mathcal{S}^{q(\beta)}=\mathcal T^{q^1,\ldots,q^k}.
		\end{equation}
		As a consequence,  for $n\geq3$, the set-valued map $q\mapsto \mathcal{S}^q$ is not Painlev\'e-Kuratowski continuous at some boundary points of $\Delta_{n-1}$. 
	\end{thm}

	\subsection{Zero-temperature limit of Gibbs-preserving semigroups and BT-stochastic matrices}
	
	Given a nonincreasing sequence $
	E=(E_1,\ldots,E_n)\in\mathbb R^n$, with $E_1\geq E_2\geq\cdots \geq E_n$,
	we say that $E$ is of type $\mu=(m_1,\ldots,m_k)$
	if
	\begin{equation}
		\begin{aligned}
			E_1&=\cdots=E_{m_1},\\
			E_{m_1+1}&=\cdots=E_{m_1+m_2},\\
			&\ \vdots\\
			E_{m_1+\cdots+m_{k-1}+1}
			&=\cdots=
			E_{m_1+\cdots+m_k},
		\end{aligned}
	\end{equation}
	with strict inequalities between consecutive blocks, $E_{m_1+\cdots+m_i} > E_{m_1+\cdots+m_i + 1}$. Thus $E$ has exactly $k$ distinct energy levels, with multiplicities $m_1,\ldots,m_k$. 
	Define the \textbf{Gibbs vector} at each inverse temperature $\beta\geq0$,
	\begin{equation}
		\label{eq:Gibbs}
		q(\beta)
		=
		\frac1{Z(\beta)}
		\bigl(
		e^{-\beta E_1},
		\ldots,
		e^{-\beta E_n}
		\bigr),
		\qquad
		Z(\beta)=\sum_{i=1}^n e^{-\beta E_i}.
	\end{equation}
	We wish to describe the $\beta\to\infty$ (zero-temperature) limit of $\mathcal{S}^{q(\beta)}$. 
	As before, it is convenient to organise matrices according to the degeneracy pattern $
	\mu=(m_1,\ldots,m_k)$.
	Note that the Gibbs vector decomposes as
	$$
	q(\beta)
	=
	\bigl(
	q^1(\beta),
	\ldots,
	q^k(\beta)
	\bigr),
	$$
	with each block $q^a(\beta)$ of length $m_a$ and equal entries.
	
	Then, for each $a=1,\ldots, k-1$,
	\begin{equation}
		\frac{\|q^a(\beta)\|_1}
		{\|q^{a+1}(\beta)\|_1}=\frac{m_a}{m_{a+1}} \exp\left(-\beta \left(E_{m_1+\cdots+{m_a}}-E_{m_1+\cdots+{m_a}+1}\right)\right)\to 0, 
	\end{equation}
	so the blocks are exponentially separated in scale. Moreover, within each block the distribution is exactly uniform. Thus  $q(\beta)$ has multiscale profile $(q^1,\ldots, q^k)$ with $q^a=\left(\frac{1}{m_a},\ldots,\frac{1}{m_a}\right)\in\Delta_{m_a-1}$ the flat vector for each $a=1,\ldots,k$.

	As temperature decreases ($\beta\to\infty$), the Gibbs distribution concentrates on the ground-state sector. More precisely, the limiting Gibbs state is uniform in the last block of dimension $m_k$, and zero everywhere else:
	\begin{equation}
		\label{eq:limit_Gibbs}
		q(\beta)
		\to
		q(\infty)
		=
		\frac1{m_k}\big(
		\underbrace{0,\ldots,0}_{n-m_k},
		\underbrace{1,\ldots,1}_{m_k}
		\big).
	\end{equation}
	The semigroup $\mathcal{S}^{q(\infty)}$ preserving the limiting Gibbs state~\eqref{eq:limit_Gibbs}
	consists of row-stochastic matrices with blocks
	\begin{equation}
		P^{k1}=0, P^{k2}=0,\ldots,  P^{k,{k-1}}=0  
	\end{equation}
	while the lower-right $m_k\times m_k$ block $P^{kk}$ is bistochastic. 		
	
	We now specialise Theorem~\ref{thm:general} to Gibbs-preserving semigroups. In this case,  transitions against the energy ordering become exponentially suppressed as $\beta\to\infty$. We define the emerging semigroup. 
	\begin{defn}
		Let $\mu=(m_1,\ldots,m_k)$ be a composition of $n$. Define  $\mathcal T_\mu$ as the semigroup of row-stochastic block-upper-triangular matrices $P=(P^{ab})_{a,b=1}^k$, where each block $P^{ab}$ has size $m_a\times m_b$, $P^{ab}=0$ if $b<a$, and 
		every diagonal block
		\begin{equation}
			P^{11},P^{22},\ldots,P^{kk}
		\end{equation}
		is column-substochastic.
	\end{defn}
	The semigroup $\mathcal{T}_{\mu}$ is a special case of $\mathcal{T}^{q^1,\ldots,q^k}$ with uniform $q^a$'s. We call the block-upper-triangular matrices in $\mathcal{T}_{\mu}$, BT-stochastic matrices. 
	\begin{rem}
		The last diagonal block $P^{kk}$ is both row-stochastic and column-substochastic. Since its entries are nonnegative, it is automatically column-stochastic, hence bistochastic.
	\end{rem}

	\begin{thm}
		\label{thm:main}
		Let $\mu=(m_1,\ldots,m_k)$ be a composition of $n$, and let
		$$
		E_1\ge\cdots\ge E_n
		$$
		be a vector of type $\mu$. Let $q(\beta)$ be the corresponding Gibbs vector~\eqref{eq:Gibbs}. Then,  \begin{equation}
			\operatorname{Lim}\mathcal{S}^{q(\beta)}=\mathcal{T}_{\mu}.
		\end{equation}
	\end{thm}
	\begin{rem}
		When $\mu=(n)$,
		$$
		E_1=\cdots=E_n,
		$$
		$q(\beta)=(\frac{1}{n},\dots,\frac{1}{n})$ for all $\beta\geq0$, and  $\mathcal{S}^{q(\beta)}=\mathcal{S}^{q(\infty)}=\mathcal{T}_\mu$ is the set of bistochastic matrices. 
		However, in general limit semigroup $\mathcal{T}_{\mu}$ is a proper subsemigroup of $\mathcal{S}^{q(\infty)}$. For example, in the nondegenerate case
		$$
		E_1>\cdots>E_n,
		\qquad
		\mu=(1,\ldots,1),
		$$
		the limiting semigroup $\mathcal T_\mu$ is the set of upper triangular row-stochastic matrices discovered by Narasimhachar and  Gour ~\cite{narasimhachar2015low}.  The difference between $\mathcal{T}_\mu$ and $\mathcal{S}^{q(\infty)}$ can be understood in terms of allowed probability flows between energy levels: $\mathcal{S}^{q(\infty)}$ inhibits only transitions from the ``lowest'' or ``ground'' energy level to the ``excited'' ones, but still allows transitions from lower to upper energy levels among the excited ones, while this is not possible within $\mathcal{T}_\mu$, which is reminiscent of the level hierarchy in the low-temperature Gibbs weights. The two conditions are therefore different whenever there is at least two excited states, counted with multiplicity, that is, when $n-m_k\geq 2$.
	\end{rem}

	Informally, if a row-stochastic matrix $P$ is in the $\beta\to\infty$ limit of the family $\left(\mathcal{S}^{q(\beta)}\right)_{\beta\geq 0}$, then there exist matrices  $P^{\beta}\in \mathcal{S}^{q(\beta)}$, such that $P^{\beta}\to P$ as $\beta\to\infty$. This implies that 
	\begin{equation}
		q(\beta)(P-I)=q(\beta)(P-P^{\beta})+\underbrace{q(\beta)(P^{\beta}-I)}_{=0}= q(\beta) (P-P^{\beta})\to 0.
	\end{equation}
	But by continuity, $q(\beta)(P-I)\to q(\infty)(P-I)$. We conclude that $q(\infty)P=q(\infty)$, namely $P\in \mathcal{S}^{q(\infty)}$. 
	
	However, not all matrices $P\in \mathcal{S}^{q(\infty)}$ can be approximated by $q(\beta)$-preserving stochastic matrices for large $\beta$. For instance, consider $n=3$,  $E_1>E_2>E_3$ and
	\begin{equation}
		P=
		\mqty(1 & 0 &0\\
		1 &0 &0 \\
		0 & 0 &1)\in \mathcal{S}^{q(\infty)}.
	\end{equation}
	Assume that there exists a family $(P^{\beta})_{\beta\geq0}$ with $P^{\beta}\in \mathcal{S}^{q(\beta)}$ such that $\|P^{\beta}-P\|\to0$, as $\beta\to\infty$. For all $\beta>0$, since $q(\beta)P^{\beta}=q(\beta)$, we have 
	\begin{equation}
		q_1(\beta)P^{\beta}_{1,1}
		+q_2(\beta)P^{\beta}_{2,1}
		+ q_3(\beta)P^{\beta}_{3,1}=q_1(\beta).
	\end{equation}
	Hence, $q_2(\beta)P^{\beta}_{2,1}\leq q_1(\beta)$, and we get to the contradiction
	$$
	P^{\beta}_{2,1}\leq \frac{q_1(\beta)}{q_2(\beta)}\to 0\neq 1=P_{2,1}.
	$$
	
	Indeed, Theorem~\ref{thm:main} claims that the zero-temperature limit of $\mathcal{S}^{q(\beta)}$ is the set of row-stochastic upper triangular matrices $\mathcal{T}_{(1,1,1)}$, a proper subsemigroup of $ \mathcal{S}^{q(\infty)}$.
	
	\subsection{Extreme points of BT-stochastic matrices}
	
	Recall that the set $\mathcal{S}^q$ of $q$-preserving stochastic matrices is compact and convex. By a classical theorem of Minkowski~\cite{brondsted2012introduction}, every compact convex set is the convex hull of its set of extreme points. For general $q$, the characterisation of the extreme points of $\mathcal{S}^q$ is non-trivial. It is known~\cite[Remark 4.5]{joe1990majorization} that, for $q\in\Delta_{n-1}$,
	the number of extreme points of $\mathcal{S}^q$ is at least $n!$ and at most $\binom{n^2}{2n-1}$. 
	If $q=\e$, then  $\mathcal{S}^q$ is the celebrated Birkhoff polytope whose extreme points are the $n!$ permutation matrices~\cite{birkhoff1946tres}.

	It is easy to see that $\mathcal{T}_{\mu}$ is also a convex set. Its extreme points can be characterised and enumerated. 
	
	\begin{prop}
		\label{prop:extreme_pts}
		Let $\mu=(m_1,\ldots, m_k)$ be a composition of $n$.  Denote by $\operatorname{ext}(\mathcal{T}_{\mu})$ the set of extreme points of $\mathcal{T}_{\mu}$. Then, 
		\begin{enumerate}
			\item $\operatorname{ext}(\mathcal{T}_{\mu})$ is the set of  block-upper-triangular $0$--$1$ matrices satisfying:
			\begin{enumerate}
				\item each row contains exactly one $1$;
				\item for every $a=1,\ldots,k$, each column of the diagonal block $P_{aa}$ contains at most one $1$.
			\end{enumerate}
			\item The number of extreme points $\operatorname{ext}(\mathcal{T}_{\mu})$ is
			\begin{equation}
				\label{eq:Laguerre_repr}
				\#\operatorname{ext}(\mathcal{T}_{\mu})=\prod_{l=1}^k\widehat{L}_{m_l}\left(m_{l+1}+\cdots+m_{k}\right),
			\end{equation}
			where $\widehat{L}_n(x):=n!L_n(-x)$, with
			\begin{equation}
				L_n(x):=[t^n]\frac{1}{1-t}e^{-\frac{tx}{1-t}}=\sum_{k=0}^{n}\binom{n}{k}\frac{(-1)^k}{k!}x^k
			\end{equation} 
			the $n$-th Laguerre polynomial, 
			and we use the convention that the empty sum is $0$.
		\end{enumerate} 
		
	\end{prop}
	
	\subsection{BT-majorisation}
	The semigroup $\mathcal{T}_{\mu}$ defines a semigroup majorisation on probability vectors~\cite{Cunden2025SemigroupMajorisation}. For $x,y\in \R^n$ we say that $x$ is $\mathcal{T}_{\mu}$-majorised by $y$, and write $x\prec^{\mathcal{T}_{\mu}}y$ if $$
	x=yP,\, \text{ for some $P\in\mathcal{T}_{\mu}$.}
	$$
	The relation $\prec^{\mathcal{T}_{\mu}}$ is a preorder on $\R^n$. We call this relation BT-majorisation. 
	
	We index vectors $x\in\R^n$ according to the block structure  $\mu=(m_1,\ldots,m_k)$, writing $x=(x^1,\ldots,x^k)$ where $x^a=(x^a_1,\ldots,x^a_{m_a})$ for all $a=1,\ldots,k$.
	
	For each $a=1,\ldots,k$,  let $(x^a)^{\downarrow}$ be the nonincreasing rearrangement of the entries $x^a$. The $\mu$-nonincreasing rearrangement of $x$ is then defined by 
	\begin{equation}\label{eq:NotationMuOrdering}
		x^{\mu,\downarrow}=\left((x^1)^{\downarrow},(x^2)^{\downarrow},\ldots, (x^k)^{\downarrow}\right).
	\end{equation}
	In words, $x^{\mu,\downarrow}$ is the vector obtained by sorting the entries of $x$ 
	within the blocks, while preserving the block order.
	
	We also define the \textbf{lexicographic} $<
	^{\mathrm{lex}}$ and \textbf{reverse lexicographic} $<
	^{\mathrm{revlex}}$ orders. For $w=(w_1,\ldots,w_k)$ and $z=(z_1,\ldots,z_k)$ real $k$-tuples, we write  $w<
	^{\mathrm{lex}}z$ if, at the first differing coordinate $i$, $w_i<z_i$; we write instead $w<
	^{\mathrm{revlex}}z$ if, at the last differing coordinate $i$, $w_i<z_i$.

	\begin{thm}
		\label{prop:monotones}
		Let $\mu$ be a composition of $n$.
		For $x,y\in\mathbb \R_+^n$, the following are equivalent:
		\begin{enumerate}
			\item $x\prec^{\mathcal T_\mu}y$;
			\item for every $\ell=1,\ldots,n$,
			$$
			\sum_{i=1}^{\ell}x_i^{\mu,\downarrow}
			\le
			\sum_{i=1}^{\ell}y_i^{\mu,\downarrow},
			$$
			with equality for $\ell=n$.
		\end{enumerate}
	\end{thm}
	
	BT-majorisation interpolates between standard majorisation (when $\mu=(n)$) and upper triangular (UT) or unordered majorisation (when $\mu=(1,\ldots,1)$).
	\subsection{The fate of monotones}
	It is known~\cite{vom2019d} that for $q$ with strictly positive components
	\begin{equation}
		\label{eq:thermo}
		x\prec^{q} y\quad\iff\quad \sum_j q_j\, g\!\left(\frac{x_j}{q_j}\right)
		\le
		\sum_j q_j\, g\!\left(\frac{y_j}{q_j}\right)
	\end{equation}
	for every continuous convex function $g$ defined on a domain containing the points
	$
	x_j/q_j
	$
	and
	$
	y_j/q_j
	$. Assume that $q=q(\beta)$ is a Gibbs vector parametrised by the inverse temperature $\beta$.
	
	Choosing
	$$
	g(t)=\frac{1}{\beta}t\log t,
	\qquad t\ge0,
	$$
	in~\eqref{eq:thermo},
	one finds, 
	\begin{equation}
		x\prec^{q(\beta)} y
		\;\Longrightarrow\;
		F_\beta(x)\le F_\beta(y),
	\end{equation}
	where
	$$
	F_\beta=U-\frac1\beta S
	$$
	is the free energy, 
	\begin{equation}
		U(x)=\sum_j x_j E_j,
		\qquad
		S(x)=-\sum_j x_j\log x_j
	\end{equation}
	with $U$ and $S$ denoting the mean energy and entropy, respectively.
	
	Choosing instead
	\begin{equation}
		g(t)
		=
		\frac{1}{\alpha-1}t^\alpha,
		\qquad
		\alpha>0,\,\, \alpha\neq 1,
	\end{equation}
	which is strictly convex on $t>0$, one obtains
	\begin{equation}
		x\prec^{q(\beta)} y
		\;\Longrightarrow\;
		\Phi_{\alpha}(x)\le \Phi_{\alpha}(y),
	\end{equation}
	where
	\begin{equation}
		\Phi_{\alpha}(x):=
		\frac{Z(\beta)^{\alpha-1}}{\alpha-1}
		\sum_j x_j^\alpha e^{-\beta(1-\alpha)E_j},
	\end{equation}
	is a  linearised relative $\alpha$-entropy.
	
	In the language of physics, if $x$ is thermomajorised by $y$ at inverse temperature $\beta$, then the free energy and the $\alpha$-entropies of $x$ do not exceed those of $y$. The following proposition shows that, in the zero-temperature limit, these entropic inequalities collapse to much simpler conditions.
	\begin{prop}
		\label{prop:fate_monotones}
		Let $x,y\in\Delta_{n-1}$ and suppose that there exists $\beta_0>0$ such that $x\prec^{q(\beta)} y$  for all $\beta>\beta_0$.  Then
		either
		$$
		U(x)<U(y),
		$$
		or
		$$
		U(x)=U(y)
		\quad\text{and}\quad
		S(x)\ge S(y).
		$$
		Moreover, for every fixed $\alpha>0$, $\alpha\neq1$, and $z\in\Delta_{n-1}$, define
		\begin{equation}
			M_a^{(\alpha)}(x)
			:=
			\sum_{j=1}^{m_a} (x_j^a)^\alpha,
			\qquad a=1,\ldots,k,
		\end{equation}
		and
		\begin{equation}
			M^{(\alpha)}(z)
			:=
			\bigl(
			M_1^{(\alpha)}(z),\ldots,M_k^{(\alpha)}(z)
			\bigr)
		\end{equation}
		Then
		\begin{equation}
			M^{(\alpha)}(x)
			\leq^{\mathrm{lex}}
			M^{(\alpha)}(y),
			\qquad \alpha>1,
			\label{eq:renyi-lex}
		\end{equation}
		whereas
		\begin{equation}
			M^{(\alpha)}(y)
			\leq^{\mathrm{revlex}}
			M^{(\alpha)}(x),
			\qquad 0<\alpha<1.
			\label{eq:renyi-revlex}
		\end{equation}
		Finally, both \eqref{eq:renyi-lex} and \eqref{eq:renyi-revlex} hold in the limit $\alpha\rightarrow 1$, namely: 
		\begin{equation}
			\label{eq:lex-revlex-condition}
			(
			\|x^1\|_1,
			\dots,
			\|x^k\|_1
			)
			\leq
			^{\mathrm{lex}}
			(
			\|y^1\|_1,
			\dots,
			\|y^k\|_1
			)
			\leq 
			^{\mathrm{revlex}}
			(
			\|x^1\|_1,
			\dots,
			\|x^k\|_1
			),
		\end{equation} 
		where 
		\[\|x^a\|_1=\sum_{j=1}^{m_a}x^a_j \quad \text{ and } \quad \|y^a\|_1=\sum_{j=1}^{m_a}y^a_j 
		\]
		for all $a=1,\ldots,k$.
	\end{prop}
	Thus, in the low-temperature limit, the mean energy is the primary monotone in the free energy, while entropy provides the subleading order within a fixed energy sector.
	Indeed, if $U(x)=U(y)$, then
	\begin{equation}
		F_\beta(y)-F_\beta(x)
		=
		\frac1\beta\bigl(S(x)-S(y)\bigr),
	\end{equation}
	so transitions preserving the mean energy can only increase entropy. Restricted to a fixed-energy subspace, the dynamics reduces to bistochastic mixing, which preserves the uniform distribution on that subspace. Decreasing entropy therefore requires population transfer to lower energy sectors.

	\subsection{BT-matrices in dimensions $n=3$ and $n=4$}
	\label{sec:limit_sets}
	
	For $n=3$, there are four possible types of degeneracies of 
	$E_1 \ge E_2 \ge E_3 $:
	\begin{table}[h]
		\centering
		\begin{tabular}{r|c|c|c}
			$\mu$ &  $\#$ blocks $k$ & energy structure & $\#\operatorname{ext}(\mathcal{T}_{\mu})$ \\
			\hline
			$(1,1,1)$ & 3 & $E_1>E_2>E_3$ & 6\\
			$(1,2)$ & 2 & $E_1>E_2=E_3$& 6\\
			$(2,1)$ & 2 & $E_1=E_2>E_3$& 7\\
			$(3)$ & 1 & $E_1=E_2=E_3$ & 6 
		\end{tabular}
	\end{table}

	We now describe the limiting semigroups $\mathcal T_\mu$.
	The pattern is always the same: blocks above the energy ordering survive, blocks below vanish, and diagonal blocks carry the column-substochastic constraints. The last diagonal block is bistochastic.
	\begin{align}
		\mathcal T_{(1,1,1)}&=
		\left\{P=
		\begin{pmatrix}
			* & * & *\\
			0 & * & * \\
			0 & 0 & * \\
		\end{pmatrix}
		\text{ row-stochastic}
		\right\},\\
		\mathcal T_{(1,2)}&=
		\left\{
		P= \begin{pmatrix}
			* & * & * \\
			0 & \square & \square\\
			0 & \square & \square
		\end{pmatrix}
		\text{ row-stochastic}
		\colon P^{22}=
		\begin{pmatrix}
			\square  & \square \\
			\square  & \square 
		\end{pmatrix}
		\text{ bistochastic}
		\right\},\nonumber\\
		\mathcal T_{(2,1)}&=
		\left\{P=
		\begin{pmatrix}
			\triangle & \triangle  & *\\
			\triangle & \triangle & *\\
			0 & 0 & 1
		\end{pmatrix}
		\text{ row-stochastic}
		\colon
		P^{11}=\begin{pmatrix}
			\triangle & \triangle \\
			\triangle & \triangle \\
		\end{pmatrix}
		\text{ column-substochastic}
		\right\},\nonumber\\
		\mathcal T_{(3)}&=\left\{P=
		\begin{pmatrix}
			\square & \square  & \square\\
			\square & \square & \square\\
			\square & \square & \square
		\end{pmatrix} \text{ 
			bistochastic}\right\}.\nonumber
	\end{align}
	
	For $n=4$, there are eight possible types of degeneracies of 
	$E_1 \ge E_2 \ge E_3 \ge E_4$:
	
	\begin{table}[h]
		\centering
		\begin{tabular}{r|c|c|c}
			$\mu$ & Number of blocks $k$ & energy structure &  $\#\operatorname{ext}(\mathcal{T}_{\mu})$ \\
			\hline
			$(1,1,1,1)$ & 4 & $E_1>E_2>E_3>E_4$ & 24\\
			$(1,1,2)$ & 3 & $E_1>E_2>E_3=E_4$ &24\\
			$(1,2,1)$ & 3 & $E_1>E_2=E_3>E_4$ &28\\
			$(1,3)$ & 2 & $E_1>E_2=E_3=E_4$ & 24\\
			$(2,1,1)$ & 3 & $E_1=E_2>E_3>E_4$ & 28\\
			$(2,2)$ & 2 & $E_1=E_2>E_3=E_4$ & 28\\
			$(3,1)$ & 2 & $E_1=E_2=E_3>E_4$ & 34\\
			$(4)$ & 1 & $E_1=E_2=E_3=E_4$ & 24
		\end{tabular}
	\end{table}
	
	The limiting semigroups $\mathcal T_\mu$ are:
	\begin{align}
			\mathcal T_{(1,1,1,1)}&=
			\left\{P=
			\begin{pmatrix}
				* & * & * & *\\
				0 & * & * & *\\
				0 & 0 & * & *\\
				0 & 0 & 0 & *
			\end{pmatrix}
			\text{ row-stochastic}
			\right\},\\
			\mathcal T_{(1,1,2)}&=
			\left\{P=
			\begin{pmatrix}
				* & * & * & *\\
				0 & * & * & *\\
				0 & 0 & \square & \square\\
				0 & 0 & \square & \square
			\end{pmatrix}
			\text{ row-stochastic}
			\colon P^{33}=
			\begin{pmatrix}
				\square  & \square \\
				\square  & \square 
			\end{pmatrix}
			\text{  bistochastic}
			\right\},\nonumber\\
			\mathcal T_{(1,2,1)}&=
			\left\{P=
			\begin{pmatrix}
				* & * & * & *\\
				0 & \triangle & \triangle & *\\
				0 & \triangle & \triangle & *\\
				0 & 0 & 0 & 1
			\end{pmatrix}
			\text{ row-stochastic}
			\colon P^{22}=
			\begin{pmatrix}
				\triangle & \triangle\\
				\triangle & \triangle
			\end{pmatrix}
			\text{ column-substochastic}
			\right\},\nonumber\\
			\mathcal T_{(1,3)}&=
			\left\{P=
			\begin{pmatrix}
				* & * & * & *\\
				0 & \square &  \square &  \square\\
				0 & \square & \square &  \square\\
				0 &  \square &  \square &  \square
			\end{pmatrix}
			\text{ row-stochastic}
			\colon P^{22}
			=			\begin{pmatrix}
				\square &  \square &  \square\\
				\square &  \square &  \square\\
				\square &  \square &  \square
			\end{pmatrix}
			\text{  bistochastic}
			\right\},\nonumber\\
			\mathcal T_{(2,1,1)}&=
			\left\{P=
			\begin{pmatrix}
				\triangle & \triangle & * & *\\
				\triangle & \triangle & * & *\\
				0 & 0 & * & *\\
				0 & 0 & 0 & 1
			\end{pmatrix}
			\text{ row-stochastic}
			\colon P^{11}=
			\begin{pmatrix}
				\triangle & \triangle\\
				\triangle & \triangle
			\end{pmatrix}
			\text{ column-substochastic}
			\right\},\nonumber\\
			\mathcal T_{(2,2)}&=
			\left\{P=
			\begin{pmatrix}
				\triangle & \triangle  & * & *\\
				\triangle  & \triangle  & * & *\\
				0 & 0 & \square & \square\\
				0 & 0 & \square & \square
			\end{pmatrix}
			\text{ row-stochastic}\colon
			\begin{aligned}
				P^{11}=&
				\begin{pmatrix}
					\triangle & \triangle\\
					\triangle & \triangle
				\end{pmatrix} \text{ column-substochastic},\\
				P^{22}=&
				\begin{pmatrix}
					\square & \square \\
					\square  & \square 
				\end{pmatrix} \text{ bistochastic}
			\end{aligned}
			\right\},\nonumber\\
			\mathcal T_{(3,1)}&=
			\left\{P=
			\begin{pmatrix}
				\triangle & \triangle & \triangle & *\\
				\triangle & \triangle & \triangle & *\\
				\triangle & \triangle & \triangle & *\\
				0 & 0 & 0 & 1
			\end{pmatrix}
			\text{ row-stochastic}
			\colon P^{11}=
			\begin{pmatrix}
				\triangle & \triangle & \triangle\\
				\triangle & \triangle & \triangle\\
				\triangle & \triangle & \triangle
			\end{pmatrix}
			\text{ column-substochastic}
			\right\},\nonumber\\
			\mathcal T_{(4)}&=\left\{P=
			\begin{pmatrix}
				\square & \square & \square & \square\\
				\square & \square & \square & \square\\
				\square & \square & \square & \square\\
				\square & \square & \square & \square
			\end{pmatrix} \text{ bistochastic}\right\}.\nonumber
		\end{align}

		\section{Proofs}
		\label{sec:proofs}	
		\begin{proof}[Proof of Theorem~\ref{thm:general}]
			Since $
			\operatorname*{Li}\mathcal{S}^{q(\beta)}
			\subseteq
			\operatorname*{Ls}\mathcal{S}^{q(\beta)}
			$,
			to show that
			$$
			\operatorname*{Lim}\mathcal{S}^{q(\beta)}
			=
			\mathcal T^{q^1,\ldots,q^k}
			$$
			it suffices to prove the two inclusions
			\begin{equation}
				\label{eq:2_inclusions}
				\operatorname*{Ls}\mathcal{S}^{q(\beta)}
				\subseteq
				\mathcal T^{q^1,\ldots,q^k},
				\qquad
				\mathcal T^{q^1,\ldots,q^k}
				\subseteq
				\operatorname*{Li}\mathcal{S}^{q(\beta)}.
			\end{equation}
			We begin with the upper limit. Suppose that
			$$
			P\in\operatorname*{Ls}\mathcal{S}^{q(\beta)}.
			$$
			Then, there exists a sequence $\beta_\ell\to\infty$ and matrices
			$P^{\beta_\ell}\in S^{q(\beta_\ell)}$
			such that $P^{\beta_\ell}\to P$.
			Since
			$$
			q(\beta_\ell)P^{\beta_\ell}
			=
			q(\beta_\ell),
			$$
			for every block index $b=1,\ldots,k$,
			\begin{equation}
				\label{eq:block-balance}
				\sum_{a=1}^k
				q^{a}(\beta_\ell)
				\bigl(P^{\beta_\ell}\bigr)^{ab}
				=
				q^{b}(\beta_\ell),
			\end{equation}
			where $q^{a}(\beta_\ell)$ and 
			$\bigl(P^{\beta_\ell}\bigr)^{ab}$ are conformable for row-by-column multiplication.
			We first show that all strictly lower-triangular blocks of $P^{\beta_{\ell}}$ vanish in the limit.
			
			Fix indices
			$$
			1\le b<a\le k.
			$$
			Since all matrix entries are nonnegative, equation~\eqref{eq:block-balance} implies
			$$
			q^{a}(\beta_\ell)
			\bigl(P^{\beta_\ell}\bigr)^{ab}
			\le
			q^{b}(\beta_\ell),
			$$
			componentwise. Recall that for $b<a$, $$\|q^b(\beta_{\ell})\|_1\ll \|q^a(\beta_{\ell})\|_1,$$ as $\ell\to\infty$. Therefore,
			$$
			\bigl(P^{\beta_\ell}\bigr)^{ab}\to0,
			$$
			and passing to the limit gives
			$$
			P^{ab}=0,
			\quad \text{for $b<a$}.
			$$
			Thus $P$ is block upper triangular.
			
			We now consider the diagonal blocks. Again from~\eqref{eq:block-balance},
			taking $b=a$, we obtain
			$$
			q^{a}(\beta_\ell)
			\bigl(P^{\beta_\ell}\bigr)^{aa}
			\le
			q^{a}(\beta_\ell).
			$$
			Therefore every diagonal block $\bigl(P^{\beta_\ell}\bigr)^{aa}$
			is $q^a(\beta_{\ell})$-subpreserving. Passing to the limit shows that each diagonal block $P^{aa}$ is $q^a$-subpreserving. We conclude that $
			P\in\mathcal T^{q^1,\ldots,q^k}$.
			This proves
			$$
			\operatorname*{Ls}\mathcal{S}^{q(\beta)}
			\subseteq
			\mathcal T^{q^1,\ldots,q^k}.
			$$

			We now prove the other inclusion. Let
			$P\in\mathcal{T}^{q^1,\ldots,q^k}$. We now construct matrices $P^\beta\in \mathcal{S}^{q(\beta)}$ such that $P^\beta\rightarrow P$. The idea is to obtain $P^\beta$ by replacing the last row of $P$ in order to enforce stationarity  $q(\beta)P^\beta=q(\beta)$. However, this replacement cannot be made directly while guaranteeing positivity of the resulting $P^\beta$, since it may happen that the other rows of $P$ send already too much probability into some column, and a non-negative last row cannot subtract that excess. This obstruction is overcome by first considering a perturbation $\widetilde{P}_\delta$ of $P$, allowing small changes throughout the matrix in order to make room for the subsequent correction, and then take a suitable $\delta=\delta(\beta)\rightarrow 0$. In particular, we choose $\delta(\beta)\in(0,1)$ tending to zero sufficiently slowly to ensure that
			\begin{align}
				\|q(\beta)-q(\infty)\|_1
				&=o(\delta(\beta)),\label{eq:deltaAss1}\\
				\left(
				\frac{(q(\beta)P)_j}{q_j(\beta)}-1
				\right)_+
				&=o(\delta(\beta)),
				\qquad j=1,\ldots,n-m_k,\label{eq:deltaAss2}
			\end{align}
			where $u_+=\max\{u,0\}$.
			Such a choice is possible since all quantities on the left tend
			to zero: the multiscale assumptions and diagonal subpreservation
			imply that each ratio in~\eqref{eq:deltaAss2} converges to a value
			at most one.
			Define the matrix $
			R=\e^Tq(\infty)$, with rows all equal to $q(\infty)=(0,\dots,0,q^k)$, and 
			\begin{equation}\label{eq:Pdelta}
				\widetilde P_{\delta(\beta)}
				=(1-\delta(\beta))P+\delta(\beta)R.
			\end{equation}
			By convexity,
			$\widetilde P_{\delta(\beta)}\in\mathcal{T}^{q^1,\ldots,q^k}$.
			Let $\widetilde p_\beta$ be its last row, and define
			\[
			p_\beta
			=
			\widetilde p_\beta+
			\frac{q(\beta)-q(\beta)\widetilde P_{\delta(\beta)}}
			{q_n(\beta)}.
			\]
			Let $P^\beta$ be obtained from $\widetilde P_{\delta(\beta)}$
			by replacing its last row with $p_\beta$. By construction,
			\[
			q(\beta)P^\beta=q(\beta),
			\qquad
			P^\beta\e^T=\e^T.
			\]
			Using $q(\infty)\widetilde P_{\delta(\beta)}=q(\infty)$,
			the triangle inequality and the fact that for any row-stochastic matrix $P$, $\|zP\|_1\leq\|z\|_1$, we obtain
			\begin{align*}
				\|p_\beta-\widetilde p_\beta\|_1
				&\leq
				\frac{1}{q_n(\beta)}
				\Bigl(
				\|q(\beta)-q(\infty)\|_1
				+
				\|(q(\beta)-q(\infty))
				\widetilde P_{\delta(\beta)}\|_1
				\Bigr)
				\\
				&\leq
				\frac{2}{q_n(\beta)}
				\|q(\beta)-q(\infty)\|_1
				=o(\delta(\beta)),
			\end{align*}
			where the last step follows from
			$q_n(\beta)\to q^k_{m_k}>0$ and~\eqref{eq:deltaAss1}.
			It remains to check nonnegativity of the modified last row. Note that the last block of the last row of~\eqref{eq:Pdelta} satisfies $(\widetilde{p}_\beta)^k\geq\delta(\beta)q^k$. Then:
			\[
			(p_\beta)^k
			=
			(\widetilde p_\beta)^k+o(\delta(\beta))
			\geq
			\delta(\beta)q^k+o(\delta(\beta)).
			\]
			Since $q^k$ is strictly positive, these entries are positive
			for sufficiently large $\beta$.
			
			For $j=1,\ldots,n-m_k$, block upper triangularity gives
			$P_{nj}=0$, while the $j$th column of $R$ is zero. Hence
			\[
			(p_\beta)_j
			=
			\frac{q_j(\beta)}{q_n(\beta)}
			\left[
			1-(1-\delta(\beta))
			\frac{(q(\beta)P)_j}{q_j(\beta)}
			\right].
			\]
			Then, using \eqref{eq:deltaAss2} 
			\[
			(p_\beta)_j
			\geq
			\frac{q_j(\beta)}{q_n(\beta)}
			\bigl[\delta(\beta)-o(\delta(\beta))\bigr]
			>0
			\]
			for sufficiently large $\beta$.
			
			Thus $P^\beta\in\mathcal{S}^{q(\beta)}$ for all sufficiently
			large $\beta$. Moreover,
			$\widetilde P_{\delta(\beta)}\to P$ and
			$p_\beta-\widetilde p_\beta\to0$, so $P^\beta\to P$.
			Since $P$ was arbitrary, we conclude that
			\[
			\mathcal{T}^{q^1,\ldots,q^k}
			\subseteq
			\operatorname{Li}_{\beta\to\infty}\mathcal{S}^{q(\beta)}.
			\]
			Combining the two inclusions~\eqref{eq:2_inclusions}
			we conclude that
			\begin{equation}
				\operatorname{Lim}\mathcal{S}^{q(\beta)}= \mathcal T^{q^1,\ldots,q^k}.
			\end{equation}

		\end{proof}

		\begin{proof}[Proof of Proposition~\ref{prop:extreme_pts}] The extreme points of $\mathcal{T}_{\mu}$ are $0$--$1$ matrices:
			\begin{equation}
				\operatorname{ext}\left(\mathcal{T}_{\mu}\right)=\left\{P\in\{0,1\}^{n\times n}\colon P\in \mathcal{T}_{\mu}\right\}.
			\end{equation}
			In fact, every matrix $P\in\{0,1\}^{n\times n}\cap \mathcal{T}_\mu$ is extremal as can be easily shown by contradiction; to prove the reverse inclusion, one needs to show that the inequality constraints defining $\mathcal{T}_\mu$ do not create additional extremal points with respect to these ones. Suppose, by contradiction, that an extreme point
			$P\in\mathcal{T}_\mu$ has $N>0$ entries strictly between $0$ and $1$.
			Let $r$ be the number of rows containing such entries, and $c$
			the number of columns of the diagonal blocks that contain such
			entries and have sum $1$.
			Each counted row or column contains at least two fractional
			entries, since its entries sum to $1$, hence
			\begin{equation}\label{eq:UpperRC}
				2r\leq N,\qquad 2c\leq N.
			\end{equation}
			We will construct a nonzero matrix $H$, supported only on the
			fractional entries of $P$, such that $P\pm\varepsilon H$ belong
			to $\mathcal{T}_\mu$ for sufficiently small $\varepsilon>0$.
			The decomposition
			\[
			P=\frac12(P-\varepsilon H)+\frac12(P+\varepsilon H)
			\]
			will then contradict the extremality of $P$.
			
			Since $H$ is supported only on the fractional entries of $P$,
			nonnegativity is preserved for sufficiently small $\varepsilon$,
			and all prescribed zeros remain unchanged.
			Preserving the row sums requires
			\[
			\sum_j H_{ij}=0.
			\]
			This gives $r$ nontrivial equations, as all other rows of $H$
			are identically zero.
			
			It remains to consider the column sums of the diagonal blocks.
			If such a column has sum strictly smaller than $1$, its
			inequality is preserved for sufficiently small $\varepsilon$.
			If its sum equals $1$, the two inequalities
			\[
			\sum_i\bigl(P^{aa}_{ij}\pm\varepsilon H^{aa}_{ij}\bigr)
			=1\pm\varepsilon\sum_i H^{aa}_{ij}\leq1
			\]
			hold simultaneously if and only if
			\[
			\sum_i H^{aa}_{ij}=0.
			\]
			Therefore, only the $c$ saturated columns containing fractional entries
			give nontrivial equations. 
			Thus these requirements give a homogeneous linear system
			with $N$ unknowns and $r+c\leq N$ equations.
			If $r+c=N$, then \eqref{eq:UpperRC} gives $2c=N$, so every fractional entry belongs
			to one of the counted columns. Consequently, the sum of
			the row equations coincides with the sum of the column
			equations, making the system linearly dependent.
			Thus, in either case, a nonzero solution $H$ exists.

			The block-upper-triangular $0$--$1$ matrices $P$  in $\mathcal{T}_{\mu}$ are exactly those that satisfy the following two conditions:
			\begin{enumerate}
				\item each row contains exactly one $1$ (row-stochasticity of $P$);
				\item for every $a=1,\ldots,k$, each column of the diagonal block $P_{aa}$ contains at most one $1$ (column-substochasticity of $P^{aa}$).
			\end{enumerate}
			
			Using the above characterisation, we now  derive the counting formula
			\begin{equation}
				\label{eq:claimed_formula}
				\#\operatorname{ext}\left(\mathcal{T}_{\mu}\right)=\prod_{a=1}^k
				\sum_{i+i'=m_a}
				i!\binom{m_a}{i}^2
				\left(m_{a+1}+\cdots+m_k\right)^{i'}.
			\end{equation}
			The Laguerre representation~\eqref{eq:Laguerre_repr} then follows immediately.
			
			Since the blocks are independent, it suffices to count, for each fixed $a$, the number of admissible matrices of size
			$$
			m_a\times (m_a+\cdots+m_k).
			$$
			
			We distinguish two types of rows.
			
			Suppose that exactly $i$ rows place their $1$ inside the diagonal block $P^{aa}$. The remaining $
			i'=m_a-i$
			rows place their $1$ strictly to the right of the diagonal block.
			There are
			$$
			\binom{m_a}{i}
			$$
			ways to choose the $i$ rows contributing to the diagonal block. Once these rows are selected, their $1$'s must occupy distinct columns among the first $m_a$ columns. Hence the number of possibilities is
			$$
			m_a(m_a-1)\cdots(m_a-i+1)
			=
			i!\binom{m_a}{i}.
			$$
			The remaining $i'$ rows may place their $1$ arbitrarily among the last
			$m_{a+1}+\cdots+m_k$
			columns, giving
			$$
			\left(m_{a+1}+\cdots+m_k\right)^{i'}
			$$
			possibilities.
			
			Therefore the total number of admissible matrices associated with the $a$-th block row is
			$$
			\sum_{i+i'=m_a}
			\binom{m_a}{i}\,\cdot
			\,i!\binom{m_a}{i}\,\cdot\,
			\left(m_{a+1}+\cdots+m_k\right)^{i'}=\sum_{i+i'=m_a}
			i!\binom{m_a}{i}^2
			\left(m_{a+1}+\cdots+m_k\right)^{i'}
			$$
			Multiplying over $a=1,\ldots,k$
			yields~\eqref{eq:claimed_formula}.
		\end{proof}

		\begin{proof}[Proof of Theorem~\ref{prop:monotones}] 
			
			Suppose first that 
			$x=yP$ 
			for some $P\in\mathcal T_\mu$.
			Write $P=(P^{ab})_{a,b=1}^k$ according to the block decomposition induced by
			$\mu=(m_1,\ldots,m_k)$.
			Since $P$ is block upper triangular,
			$$
			x^a
			=
			y^aP^{aa}
			+
			z^a,
			\qquad
			a=1,\ldots,k,
			$$
			where
			$$
			z^a:=\sum_{b<a}y^bP^{ba}.
			$$
			Since the diagonal block $P^{aa}$ is doubly-substochastic, weak majorisation theory implies
			$$
			\sum_{i=1}^r
			\bigl(y^aP^{aa}\bigr)_i^\downarrow
			\le
			\sum_{i=1}^r
			(y^a)_i^\downarrow,\quad\text{for all $r=1,\ldots,m_a$.}
			$$
			Moreover, since $z^a$ has nonnegative entries,
			$$
			\sum_{i=1}^r
			(x^a)_i^\downarrow
			=
			\sum_{i=1}^r
			\bigl(y^aP^{aa}+z^a\bigr)_i^\downarrow
			\le
			\sum_{i=1}^r
			\bigl(y^aP^{aa}\bigr)_i^\downarrow
			+
			\|z^a\|_1,
			$$
			because the sum of the $r$ largest entries of a nonnegative vector increases
			by at most the total mass that is added.
			Combining the two inequalities yields
			\begin{equation}
				\label{eq:block}
				\sum_{i=1}^r
				(x^a)_i^\downarrow
				\le
				\sum_{i=1}^r
				(y^a)_i^\downarrow
				+
				\|z^a\|_1,
				\qquad
				r=1,\ldots,m_a.
			\end{equation}
			
			Now let
			$$
			\ell
			=
			m_1+\cdots+m_{a-1}+r,
			\qquad
			1\le r\le m_a.
			$$
			Then
			$$
			\sum_{i=1}^{\ell}x_i^{\mu,\downarrow}
			=
			\sum_{c<a}\|x^c\|_1
			+
			\sum_{i=1}^r(x^a)_i^\downarrow.
			$$
			Applying~\eqref{eq:block},
			\begin{equation}
				\label{eq:intermediate_transf}
				\sum_{i=1}^{\ell}x_i^{\mu,\downarrow}
				\le
				\sum_{c<a}\|x^c\|_1
				+
				\sum_{i=1}^r(y^a)_i^\downarrow
				+
				\|z^a\|_1. 
			\end{equation}
			It therefore remains to compare the total masses of the preceding blocks.
			For every block $c$,
			$$
			\|x^c\|_1
			=
			\|y^cP^{cc}\|_1+\|z^c\|_1.
			$$
			Since $P$ is row-stochastic,
			$$
			\|y^cP^{cc}\|_1
			=
			\|y^c\|_1
			-
			\sum_{d>c}\|y^cP^{cd}\|_1.
			$$
			Hence
			$$
			\|x^c\|_1
			=
			\|y^c\|_1
			-
			\sum_{d>c}\|y^cP^{cd}\|_1
			+
			\|z^c\|_1.
			$$
			Summing over $c<a$ gives
			\begin{equation}
				\label{eq:mass_transf}
				\sum_{c<a}\|x^c\|_1
				=
				\sum_{c<a}\|y^c\|_1
				-
				\sum_{c<a}\sum_{d>c}\|y^cP^{cd}\|_1
				+
				\sum_{c<a}\|z^c\|_1. 
			\end{equation}
			Observe that
			$$
			\|z^c\|_1
			=
			\sum_{b<c}\|y^bP^{bc}\|_1,
			$$
			since all matrices and vectors are nonnegative.
			Hence
			$$
			\sum_{c<a}\|z^c\|_1
			=
			\sum_{b<c<a}\|y^bP^{bc}\|_1.
			$$
			On the other hand,
			$$
			\sum_{c<a}\sum_{d>c}\|y^cP^{cd}\|_1
			=
			\sum_{c<d<a}\|y^cP^{cd}\|_1
			+
			\|z^a\|_1
			+
			\sum_{\substack{c<a\\ d>a}}\|y^cP^{cd}\|_1.
			$$
			The first term on the right-hand side is precisely
			$\sum_{c<a}\|z^c\|_1$, and it accounts for all transport
			between the first $a-1$ blocks.

			Substituting this identity into~\eqref{eq:mass_transf} yields
			$$
			\sum_{c<a}\|x^c\|_1
			+\|z^a\|_1
			=
			\sum_{c<a}\|y^c\|_1
			-
			\sum_{\substack{c<a\\ d>a}}
			\|y^cP^{cd}\|_1.
			$$
			Since the subtracted terms on the right-hand side are nonnegative,
			$$
			\sum_{c<a}\|x^c\|_1
			+\|z^a\|_1
			\le
			\sum_{c<a}\|y^c\|_1.
			$$
			Therefore, the inequality~\eqref{eq:intermediate_transf} becomes
			$$
			\sum_{i=1}^{\ell}x_i^{\mu,\downarrow}
			\le
			\sum_{c<a}\|y^c\|_1
			+
			\sum_{i=1}^r(y^a)_i^\downarrow
			=
			\sum_{i=1}^{\ell}y_i^{\mu,\downarrow}.
			$$
			
			Since $\ell$ was arbitrary, we conclude that
			$$
			\sum_{i=1}^{\ell}x_i^{\mu,\downarrow}
			\le
			\sum_{i=1}^{\ell}y_i^{\mu,\downarrow},
			\qquad
			\ell=1,\ldots,n.
			$$
			
			Finally, $P$ is row-stochastic, and therefore
			$$
			\sum_{i=1}^n x_i
			=
			\sum_{i=1}^n y_i.
			$$
			Thus equality holds for $\ell=n$, completing the proof of one implication.

			Conversely, assume that
			\begin{equation}\label{eq:InequalitiesProof}
				\sum_{i=1}^\ell x_i^{\mu,\downarrow}
				\le
				\sum_{i=1}^\ell y_i^{\mu,\downarrow},
				\qquad
				\ell=1,\ldots,n,
			\end{equation}
			with equality at $\ell=n$. 
			We construct $P\in\mathcal T_\mu$ such that $x=yP$.

			We proceed block by block. Set $P^{ab}=0$ for $a>b$. For $a=1,\dots,k$, let
			\begin{equation}\label{eq:MassDeficit}
				D_a:=\sum_{b=1}^{a-1}\bigl(\|y^b\|_1-\|x^b\|_1\bigr),
			\end{equation}
			denote the deficit of mass of the first $a-1$ blocks. The inequalities
			\eqref{eq:InequalitiesProof} imply $D_a\ge0$ and
			\begin{equation}\label{eq:generalInequality}
				\sum_{j=1}^{\ell}(x^a)^\downarrow_j
				\le D_a+\sum_{j=1}^{\ell}(y^a)^\downarrow_j,
				\qquad \ell=1,\dots,m_a.
			\end{equation}
			By the weak-majorisation theorem~\cite{marshall1979inequalities},
			there exists a doubly substochastic matrix $P^{aa}$ such that
			\begin{equation}\label{eq:additionalTerm}
				x^a=(y^a+D_av^a)P^{aa},
			\end{equation}
			where $v^a$ denotes a probability vector supported in the block $a$, with all its mass concentrated on the position of one of the maximal components of $y^a$.
			Moreover, equality of the total masses gives
			$\|y^k+D_kv^k\|_1=\|x^k\|_1$, so $P^{kk}$ can be chosen bistochastic.
			
			We construct the off-diagonal blocks inductively, maintaining
			nonnegativity and partial row sums at most one. Since $D_1=0$,
			the first block satisfies $x^1=y^1P^{11}$ and establishes the
			initial step. Suppose that the first $a-1$ block columns have
			been constructed with these properties and satisfy
			\[
			x^b=\sum_{c=1}^{b}y^cP^{cb}
			\qquad \text{for } b<a.
			\]
			The deficit of mass \eqref{eq:MassDeficit} can then be written as 
			\[
			D_a=\sum_{b=1}^{a-1}\left(\|y^b\|_1-\bigg\|\sum_{c=1}^{b}y^cP^{cb}\bigg\|_1\right)=\sum_{b=1}^{a-1}\sum_{i=1}^{m_b}y_i^b r^{ba}_i,
			\]
			where we introduced the column vector $r^{ba}$ with components
			\[
			r_i^{ba}:=1-\sum_{c=b}^{a-1}\sum_{j=1}^{m_c}P_{ij}^{bc}\ge0,
			\qquad b<a,
			\]
			keeping track of the unused fraction in the $i$-th row of the blocks $P^{bb},P^{b,b+1},\dots,P^{b,a-1}$.
			Define
			\[
			P_{ij}^{ba}:=r_i^{ba}(v^aP^{aa})_j,\qquad b<a.
			\]
			These entries are nonnegative and, by \eqref{eq:additionalTerm},
			\[
			\sum_{b=1}^{a}y^bP^{ba}
			=y^aP^{aa}+D_av^aP^{aa}=x^a.
			\]
			Furthermore, row-substochasticity of $P^{aa}$ gives
			\[
			\sum_{j=1}^{m_a}P_{ij}^{ba}
			=r_i^{ba}\sum_{j=1}^{m_a}(v^aP^{aa})_j\le r_i^{ba}.
			\]
			Thus the partial row sums remain at most one, including those
			of the new diagonal block, completing the induction.
			At $a=k$, bistochasticity of $P^{kk}$ makes the last inequality
			an equality, exhausting all remaining row fractions. Since the
			rows of $P^{kk}$ also sum to one, the resulting matrix belongs
			to $\mathcal T_\mu$ and satisfies $x=yP$.
		\end{proof}
		
		\begin{proof}[Proof of Proposition~\ref{prop:fate_monotones}]
			Assume that there exists $\beta_0$, such that $
			x\prec^{q(\beta)} y$ for all $\beta>\beta_0$. Then, if $\beta>\beta_0$,
			\begin{equation}
				F_{\beta}(x)=U(x)-\frac1\beta S(x)
				\le
				U(y)-\frac1\beta S(y)=F_{\beta}(y).
			\end{equation}
			Hence, either $
			U(x)<U(y)$,
			or
			$U(x)=U(y) $ but 
			$S(x)\ge S(y)$.
			
			Consider now the inequalities
			$$
			\Phi_{\alpha}(x)\le \Phi_{\alpha}(y) \Leftrightarrow \frac{1}{\alpha-1}\sum_j x_j^\alpha e^{-\beta(1-\alpha)E_j}\leq  \frac{1}{\alpha-1}\sum_j y_j^\alpha e^{-\beta(1-\alpha)E_j}
			$$
			for all $\beta>\beta_0$ and $\alpha>0$, $\alpha\neq 1$. 
			If we denote with $(\mathcal{E}_1,\dots,\mathcal{E}_k)$ the distinct energy levels of $(E_1,\dots,E_n)$, the inequalities can be rewritten as
			\begin{equation}
				\sum_{a=1}^k
				M_a^{(\alpha)}(x)
				e^{\beta(\alpha-1)\mathcal E_a}
				\leq
				\sum_{a=1}^k
				M_a^{(\alpha)}(y)
				e^{\beta(\alpha-1)\mathcal E_a},
				\qquad \alpha>1,
				\label{eq:renyi-alpha-greater}
			\end{equation}
			\begin{equation}
				\sum_{a=1}^k
				M_a^{(\alpha)}(x)
				e^{\beta(\alpha-1)\mathcal E_a}
				\geq
				\sum_{a=1}^k
				M_a^{(\alpha)}(y)
				e^{\beta(\alpha-1)\mathcal E_a},
				\qquad 0<\alpha<1.
				\label{eq:renyi-alpha-smaller}
			\end{equation}
			
			Suppose first that $\alpha>1$. If
			$M^{(\alpha)}(x)\neq M^{(\alpha)}(y)$, let $r$ be the first index
			for which
			\[
			M_r^{(\alpha)}(x)\neq M_r^{(\alpha)}(y).
			\]
			Subtracting the two sides of
			\eqref{eq:renyi-alpha-greater} and dividing by
			$e^{\beta(\alpha-1)\mathcal E_r}$ gives
			\[
			M_r^{(\alpha)}(x)-M_r^{(\alpha)}(y)
			+
			\sum_{a>r}
			\left(
			M_a^{(\alpha)}(x)-M_a^{(\alpha)}(y)
			\right)
			e^{\beta(\alpha-1)(\mathcal E_a-\mathcal E_r)}
			\leq0.
			\]
			Since $\mathcal E_a<\mathcal E_r$ for $a>r$, all terms in the sum
			vanish as $\beta\to\infty$. Therefore
			\[
			M_r^{(\alpha)}(x)-M_r^{(\alpha)}(y)\leq0.
			\]
			By the choice of $r$ the two quantities are different, so in fact
			\[
			M_r^{(\alpha)}(x)<M_r^{(\alpha)}(y),
			\]
			which proves
			\[
			M^{(\alpha)}(x)<^{\mathrm{lex}}M^{(\alpha)}(y).
			\]
			Allowing equality of the two vectors $M^{(\alpha)}(x),M^{(\alpha)}(y)$ gives
			\eqref{eq:renyi-lex}. The case for $0<\alpha<1$ is analogous and leads to~\eqref{eq:renyi-revlex}.
			
			It remains to prove~\eqref{eq:lex-revlex-condition}. Take any
			sequence $\beta_\ell\to\infty$. By compactness, after passing to a subsequence we may assume that $P_{\beta_\ell}\rightarrow P$.
			By Theorem~\ref{thm:main}, $
			P\in\mathcal T_\mu$.
			Since $x=yP_{\beta_\ell}$ for every $\ell$, passing to the limit gives $x=yP$ and hence $
			x\prec_{\mathcal T_\mu}y$.
			The characterisation of BT-majorisation then implies
			\[
			\sum_{i=1}^{\ell}x_i^{\mu,\downarrow}
			\leq
			\sum_{i=1}^{\ell}y_i^{\mu,\downarrow},
			\qquad \ell=1,\ldots,n,
			\]
			with equality for $\ell=n$. Evaluating these inequalities at the
			endpoints of the blocks,
			\[
			\ell=m_1+\cdots+m_r,
			\]
			gives
			\[
			\sum_{a=1}^r\|x^a\|_1
			\leq
			\sum_{a=1}^r\|y^a\|_1,
			\qquad r=1,\ldots,k,
			\]
			with equality at $r=k$. The lexicographic consequences follow immediately.	
		\end{proof}
		
		\section*{Statements and Declarations}
		
		\subsection*{Funding} 
		This research is supported by Gruppo Nazionale di Fisica Matematica GNFM-INdAM, Istituto Nazionale di Fisica Nucleare INFN through the project QUANTUM, and financially supported by PNRR MUR project PE0000023-NQSTI, PNRR MUR project CN00000013 ‘Italian National Centre on HPC, Big Data
		and Quantum Computing’, PRIN 2022 project 2022TEB52W-PE1-
		`The charm of integrability: from nonlinear waves to random matrices', the project ``Patto territoriale sistema universitario pugliese'', the start-up grant associated with the Nanyang Assistant Professorship awarded to Nelly Ng at Nanyang Technological University, Singapore, and Taighde Éireann -- Research Ireland under Grant No.~IRCLA/2022/3922 and GOIPD/2026/3505.
		
		\subsection*{Competing interests}
		The authors have no relevant financial or non-financial interests to disclose.

		\subsection*{Declaration on the use of AI} 
		AI assistance was used for general editing  and  proofreading. 
		\bibliographystyle{alpham}
		\bibliography{References}
		
	\end{document}